\documentclass[journal]{IEEEtran}
\usepackage{cite}
\usepackage{amsmath,amssymb,amsfonts}
\usepackage{algorithmic}
\usepackage{graphicx}
\usepackage{textcomp}
\usepackage[nopostdot,acronym,shortcuts,nonumberlist]{glossaries}
\usepackage[caption=false,font=footnotesize]{subfig}
\usepackage{algorithm}
\usepackage{bbm}
\usepackage{bm}

\def\BibTeX{{\rm B\kern-.05em{\sc i\kern-.025em b}\kern-.08em
    T\kern-.1667em\lower.7ex\hbox{E}\kern-.125emX}}

\DeclareMathOperator*{\argmax}{arg\,max}

\newtheorem{proposition}{Proposition}
\newtheorem{assumption}{Assumption}

\begin{document}

\title{Energy-Aware Two-Sided Learning for Dynamic Matching Games in Mobile Crowdsensing}

\author{Sumedh~J.~Dongare,~\IEEEmembership{Student Member,~IEEE,}
        Anja~Klein,~\IEEEmembership{Member,~IEEE,}
        and~Andrea~Ortiz,~\IEEEmembership{Member,~IEEE}%
\thanks{S. J. Dongare and A. Klein are with the Communications Engineering Lab, Technical University of Darmstadt, Landgraf-Georg-Strasse 4, 64283 Darmstadt, Germany (e-mail: \{s.dongare, a.klein\}@nt.tu-darmstadt.de).}%
\thanks{A. Ortiz is with the Institute of Telecommunications, Vienna University of Technology, Austria (e-mail: andrea.ortiz@tuwien.ac.at).}%
\thanks{Corresponding author: Sumedh J. Dongare (e-mail: s.dongare@nt.tu-darmstadt.de).}%
\thanks{This work has been funded by the German Research Foundation (DFG) as a part of the project C1 and B3 within the Collaborative Research Center (CRC) 1053 - MAKI (Nr. 210487104) and has been supported by the German Federal Ministry of Research, Technology and Space (BMFTR) project Open6GHub+ under grant 16KIS2407, by DAAD with funds from BMFTR under grant 57817830, and by the LOEWE Center emergenCITY under grant LOEWE/1/12/519/03/05.001(0016)/72.
The work of Andrea Ortiz was funded by the Vienna Science and Technology Fund WWTF under grant 10.47379/VRG23002.}}

\markboth{Preprint}%
{Dongare \MakeLowercase{\textit{et al.}}: Energy-Aware Two-Sided Learning for Dynamic Matching Games in Mobile Crowdsensing}

\maketitle

\begin{abstract}
Mobile crowdsensing (MCS) is a promising enabler of Sensing-as-a-Service (SaaS) for next generation networks (NGNs), where sensing, communication, and computing are jointly considered as on-demand services. In MCS, mobile units (MUs) collect and deliver sensing data to data requesters (DRs) via a mobile crowdsensing platform (MCSP) in exchange for monetary incentives. After sensing tasks are announced, MUs strategically select tasks to maximize their long-term utility while accounting for energy and time costs, whereas the MCSP assigns tasks to maximize its own service revenue and data quality. A fundamental challenge arises from the lack of prior knowledge of MUs’ sensing qualities and task efforts, as well as the energy limitations of battery-powered devices, which directly impacts service availability and reliability in SaaS for NGNs. To address these challenges, we formulate the interaction between MUs and the MCSP as a dynamic two-sided matching game under incomplete information, explicitly incorporating energy constraints. We propose Energy-aware Two-Sided Learning (ETSL), a fully decentralized and lightweight learning framework in which MUs locally learn task proposal strategies, while the MCSP learns the data quality of participating MUs to devise task assignment strategy. ETSL jointly enables MUs' energy-aware task proposals and MCSP's adaptive task assignment, considering their individual preferences to maximize their net revenues. Simulation results demonstrate that ETSL significantly improves MU and MCSP profits and overall energy efficiency, highlighting its effectiveness as a scalable and sustainable SaaS solution for NGN.
\end{abstract}


\newacronym{sus}{SUS}{Status Update System}
\newacronym{mcs}{MCS}{Mobile Crowdsensing}
\newacronym{mu}{MU}{mobile unit}
\newacronym{mcsp}{MCSP}{mobile crowdsensing platform}
\newacronym{dr}{DR}{data requester}
\newacronym{cpu}{CPU}{Central Processing Unit}
\newacronym{psnr}{PSNR}{Peak Signal-to-Noise Ratio}
\newacronym{aoi}{AoI}{Age Of Information}
\newacronym{iot}{IoT}{Internet of Things}
\newacronym{dt}{DT}{Digital twin}
\newcommand{\timeindex}{t}
\newcommand{\timehorizon}{T}

\newcommand{\MUindex}{k}
\newcommand{\setOfMUs}{\mathcal{K}^\mathrm{MU}}
\newcommand{\MUwithIndex}[1]{\mathrm{MU}_{#1}}
\newcommand{\MUk}{\MUwithIndex{\MUindex}}
\newcommand{\energySensitivity}{{\color{red} REMOVE ME}}
\newcommand{\effortMargin}{\beta}

\newcommand{\setOfTaskTypes}{\mathcal{Z}}
\newcommand{\numberOfTaskTypes}{Z}
\newcommand{\taskTypeIndex}{z}
\newcommand{\taskTypeZ}{\taskTypeIndex}
\newcommand{\taskIndex}{n}
\newcommand{\setOfTasks}{\mathcal{A}}
\newcommand{\taskWithIndex}[1]{{a}_{#1}}
\newcommand{\taskTypeWithIndex}[1]{\mathrm{a}_{#1}}
\newcommand{\taskn}{\taskTypeWithIndex{\taskIndex}}
\newcommand{\deadlineOfTask}{\tau_{\taskTypeIndex}^{\mathrm{max}}}
\newcommand{\taskSize}{s_\taskTypeIndex}

\newcommand{\setOfTasksWithType}[1]{\setOfTasks_{{#1},\timeindex}}
\newcommand{\setOfTasksWithTypeZ}{\setOfTasksWithType{\taskTypeIndex}}
\newcommand{\mappingFunction}{g_\timeindex}

\newcommand{\MUTaskAndTimeIndex}{_{\MUindex,\taskIndex,\timeindex}}
\newcommand{\MUTaskIndex}{_{\MUindex,\taskIndex}}

\newcommand{\elementOfAssignmentMatrixWithIndex}[1]{x_{#1}}
\newcommand{\elementOfAssignmentMatrix}{\elementOfAssignmentMatrixWithIndex{\MUTaskAndTimeIndex}}
\newcommand{\assignmentMatrix}{\mathbf{X}_{\timeindex}}

\newcommand{\paymentFunction}{P^{\mathrm{effort}}}
\newcommand{\costFunction}{c_k^{\mathrm{effort}}}

\newcommand{\sensingTime}{\tau^{\mathrm{sense}}\MUTaskAndTimeIndex}
\newcommand{\expectedSensingTime}{\bar{\tau}^{\mathrm{sense}}_{\MUindex,\taskTypeIndex}}
\newcommand{\communicationTime}{\tau^{\mathrm{comm}}\MUTaskAndTimeIndex}
\newcommand{\computationTime}{\tau^{\mathrm{comp}}\MUTaskAndTimeIndex}
\newcommand{\expectedCommunicationTime}{\bar{\tau}^{\mathrm{comm}}_{\MUindex,\taskTypeIndex}}
\newcommand{\totalEnergy}{E\MUTaskAndTimeIndex}
\newcommand{\totalHarvestedEnergy}{E^{\text{h}}_{\MUindex,\timeindex}}
\newcommand{\totalTime}{\tau\MUTaskAndTimeIndex}
\newcommand{\expectedUtilityTotalTimeWithIndex}[1]{\bar{\tau}^{\mathrm{MU}}_{#1}}
\newcommand{\expectedTotalTime}{\expectedUtilityMUWithIndex{\MUindex,\taskIndex}}

\newcommand{\txPower}{p_k^\mathrm{comm}}
\newcommand{\sensePower}{p_{k,n}^\mathrm{sense}}
\newcommand{\compPower}{p_k^\mathrm{comp}}

\newcommand{\payment}{P\MUTaskAndTimeIndex}
\newcommand{\MUpaymentProposal}{\hat{P}_{\MUindex,\taskTypeIndex}}

\newcommand{\utilityMU}{U^{\mathrm{MU}}_{\MUindex,\taskIndex,\timeindex}}
\newcommand{\expectedUtilityMUWithIndex}[1]{\hat{U}^{\mathrm{MU}}_{#1}}
\newcommand{\estimatedexpectedUtilityMUWithIndex}[1]{\Tilde{U}^{\mathrm{MU}}_{#1}}
\newcommand{\expectedUtilityMU}{\expectedUtilityMUWithIndex{\MUindex,\taskTypeIndex}}

\newcommand{\utilityTask}{U^{\mathrm{MCSP}}_{\MUindex,\taskIndex,\timeindex}}
\newcommand{\expectedUtilityTaskWithIndex}[1]{\hat{U}^{\mathrm{MCSP}}_{#1}}
\newcommand{\estimatedexpectedUtilityTaskWithIndexMCSP}[1]{\Tilde{U}^{\mathrm{MCSP}}_{#1}}
\newcommand{\expectedUtilityTask}{\expectedUtilityTaskWithIndex{\MUindex,\taskTypeIndex}}

\newcommand{\rewardTaskCompletion}{w_{\taskTypeIndex,\timeindex}}
\newcommand{\rewardPerTaskType}{w_{\taskTypeIndex}}
\newcommand{\MUpreferenceWithIndex}[1]{\succeq^{\mathrm{MU}}_{#1}}
\newcommand{\MUpreference}{\MUpreferenceWithIndex{\MUindex}}
\newcommand{\MUpreferenceTimeDependent}{\MUpreferenceWithIndex{\MUindex, \timeindex}}
\newcommand{\MUstrictPreferenceWithIndex}[1]{\succ^{\mathrm{MU}}_{#1}}
\newcommand{\TaskpreferenceWithIndex}[1]{\succeq^{\mathrm{MCSP}}_{#1}}
\newcommand{\Taskpreference}{\TaskpreferenceWithIndex{\taskTypeIndex}}

\newcommand{\socialWelfare}{U^{\mathrm{SW}}_{\timeindex}(\assignmentMatrix)}

\newcommand{\stableTaskForMUk}{\taskWithIndex{k}^{\mathrm{stable}}}
\newcommand{\stableExpectedUtilityMU}{\bar{U}^{\mathrm{MU,stable}}_{\MUindex}}

\newcommand{\estimatedUtilityWithIndex}[1]{\hat{U}_{#1}}

\newcommand{\sensingOfferWithIndex}[1]{O_{#1}^{\taskTypeZ}}
\newcommand{\sensingOffer}{\sensingOfferWithIndex{\MUindex,\timeindex}}
\newcommand{\sensingAccepted}{\bar{O}_{\MUindex,\timeindex}}

\newcommand{\freeSensingAdjustmentPara}{\epsilon^{\mathrm{a}}}
\newcommand{\freeSensingStopPara}{\epsilon^{\mathrm{e}}}
\newcommand{\freeSensingSensitivityPara}{\epsilon^{\mathrm{s}}}

\newcommand{\umax}{\Delta\mathrm{U}}

\section{Introduction}
\label{sec:Introduction}
\IEEEPARstart{N}{ext} generation wireless networks are envisioned to provide sensing-as-a-service (SaaS) as a native capability, complementing communication and computing to enable intelligent and context-aware applications \cite{6G_SaaS_2024}.
To realize this vision, scalable sensing infrastructures are needed to provide measurement updates on demand and support distributed decision-making across heterogeneous sensing, communication, and computing resources.
In this context, \gls{mcs} is a promising distributed sensing framework for SaaS, where \glspl{mu}, such as smartphones, vehicles, wearables, and other mobile IoT devices, are incentivized to collect and deliver sensing data through a \gls{mcsp} in exchange for monetary compensation~\cite{Hu2023incentives}.
Since MUs are autonomous, heterogeneous, and typically battery operated, \gls{mcs} naturally introduces challenges related to incentive-, quality-, and energy-aware resource orchestration.

An \gls{mcs} system consists of three stakeholders, namely, a \gls{dr}, a \gls{mcsp}, and \glspl{mu}.
When the \gls{dr} requires sensing data from a target area, it conveys this request to the \gls{mcsp} and offers a payment for the requested sensing service.
The \gls{mcsp} converts the request into sensing tasks, publishes them to the available \glspl{mu}, and uses part of the \gls{dr}'s payment to recruit suitable participants.
The \glspl{mu} select tasks according to their own preferences and submit task proposals, including their desired payments, to the \gls{mcsp}.
In this framework, payments act as incentive signals that connect sensing demand from \glspl{dr} with distributed sensing supply from the \glspl{mu}.
From the NGN perspective, the \gls{mcsp} can be viewed as a platform-side orchestration entity, potentially implemented at the network edge or in the cloud, that coordinates task dissemination, sensing-result collection, participant selection, and quality-aware service provisioning.

The \glspl{mu} must carefully utilize their available energy while maximizing their individual utilities.
As a result, the task proposal strategy of each \gls{mu} depends on its preferences, task execution effort, expected payment, and current battery state.
Similarly, the \gls{mcsp} performs task assignments to proposing \glspl{mu} in order to improve its service utility, revenue, and the overall quality of the sensing results.
Therefore, \gls{mcs}-based SaaS requires the joint design of participant-side energy-aware decision-making and platform-side service orchestration.

The \gls{mcs} system can be modeled as a two-sided \emph{matching game}~\cite{Gu2015matching}, since the \gls{mcsp} and the \glspl{mu} hold individual preferences over task assignments and task proposals, respectively. The solution to this game aims to find task-\gls{mu} combinations that are mutually beneficial for the \gls{mcsp} and the \glspl{mu}. However, in contrast to traditional two-sided matching games, the availability of the \glspl{mu} changes over time depending on the amount of energy available in their batteries. This temporal variation transforms the problem into a more complex \emph{dynamic} matching game~\cite{SVKadam_dynamicMatching2018}.

Solving this dynamic matching game optimally would require complete information about the \gls{mcs} system, including future task arrivals, the battery state of every \gls{mu} over time, task execution efforts, and the preferences of both the \glspl{mu} and the \gls{mcsp}. Such information is generally unavailable in practical \gls{mcs} systems. Even if complete information were available, the optimal solution would exhibit poor scalability, with complexity growing exponentially in the number of \glspl{mu}, sensing tasks, and the considered time horizon. Thus, a key challenge is to design scalable and lightweight task proposal and task assignment algorithms under incomplete information and endogenous \gls{mu} availability.

In the literature, many works assume complete information about the \gls{mcs} system and optimize task allocation for latency minimization~\cite{Subband_Fu_2024}, quality maximization~\cite{Hybrid_Liu_2024}, and coverage maximization~\cite{Hybrid_Ramachandran_2024}.
The authors in~\cite{Distributed_Wang_2024} optimize task proposal strategies by considering the preferences of the \glspl{mu}.
However, the assumption of complete information limits the applicability of these approaches in realistic and dynamic \gls{mcs} deployments.
To address uncertainty, reinforcement learning has been applied to optimize task assignment~\cite{Dongare_EHMCS_2022} and task proposal~\cite{Dongare_Globecom_2023, Decentralized_Bernd_2024} strategies separately.
Nevertheless, in practical \gls{mcs}-based SaaS systems, task proposal and task assignment are tightly coupled and should be jointly considered.

Our previous work~\cite{Dongare_TSL_2024_ICC} studied decentralized task proposal by the \glspl{mu} and learning-based task assignment by the \gls{mcsp} under incomplete information, assuming unlimited \gls{mu} energy and persistent \gls{mu} availability.
In reality, the availability of mobile sensing resources is directly affected by their battery dynamics.
This fundamentally changes the problem by introducing an intertemporal coupling between learning, energy consumption, and participant availability, making the approach in~\cite{Dongare_TSL_2024_ICC} inapplicable to realistic energy-constrained scenarios.
In particular, endogenous \gls{mu} availability requires each task proposal strategy to explicitly account for both the current battery state and the time-varying competition among \glspl{mu}.
Although previous works have contributed significantly to task proposal and task assignment in \gls{mcs}, designing an efficient and scalable solution that jointly addresses incomplete information at the \gls{mcsp} and at the \glspl{mu}, while accounting for dynamic availability, remains an open problem.

The main contributions of this work are as follows.
\begin{itemize}
    \item We propose a novel fully decentralized learning framework, termed Energy-aware Two-Sided Learning (ETSL), for dynamic two-sided matching in energy-constrained \gls{mcs}.
    ETSL consists of two components: i) an Energy-aware Task Proposal algorithm (ETP), which enables each \gls{mu} to locally learn a customized task proposal strategy, and ii) a Task Assignment algorithm, which enables the \gls{mcsp} to learn the sensing quality of participating \glspl{mu} over time and adapt its assignment decisions accordingly.
    Using the proposed ETSL, the task proposal problem of the MUs and task assignment problem of the MCSP is solved simultaneously.
    \item We provide a convergence analysis of the proposed approach using the well-known alternating-freeze approach to demonstrate that the proposed ETSL solution converges to a stable solution which maximizes the individual utilities of the MUs as well as the MCSP.
    \item Using complexity analysis, we show that ETSL is a lightweight and scalable mechanism for energy-, incentive-, and quality-aware sensing-resource orchestration in \gls{mcs}-based SaaS systems for future NGNs.
    \item Extensive simulations demonstrate that ETSL is scalable and achieves social welfare within $7.6\%$ of the optimal value iteration benchmark. Moreover, its energy consumption converges to that of the benchmark, while collisions are reduced by at least $9.1\%$ compared with the considered baselines.
\end{itemize}

\section{System model}
\label{sec:system_model}
\begin{figure}[t]
\centering
\includegraphics[width=0.45\textwidth]{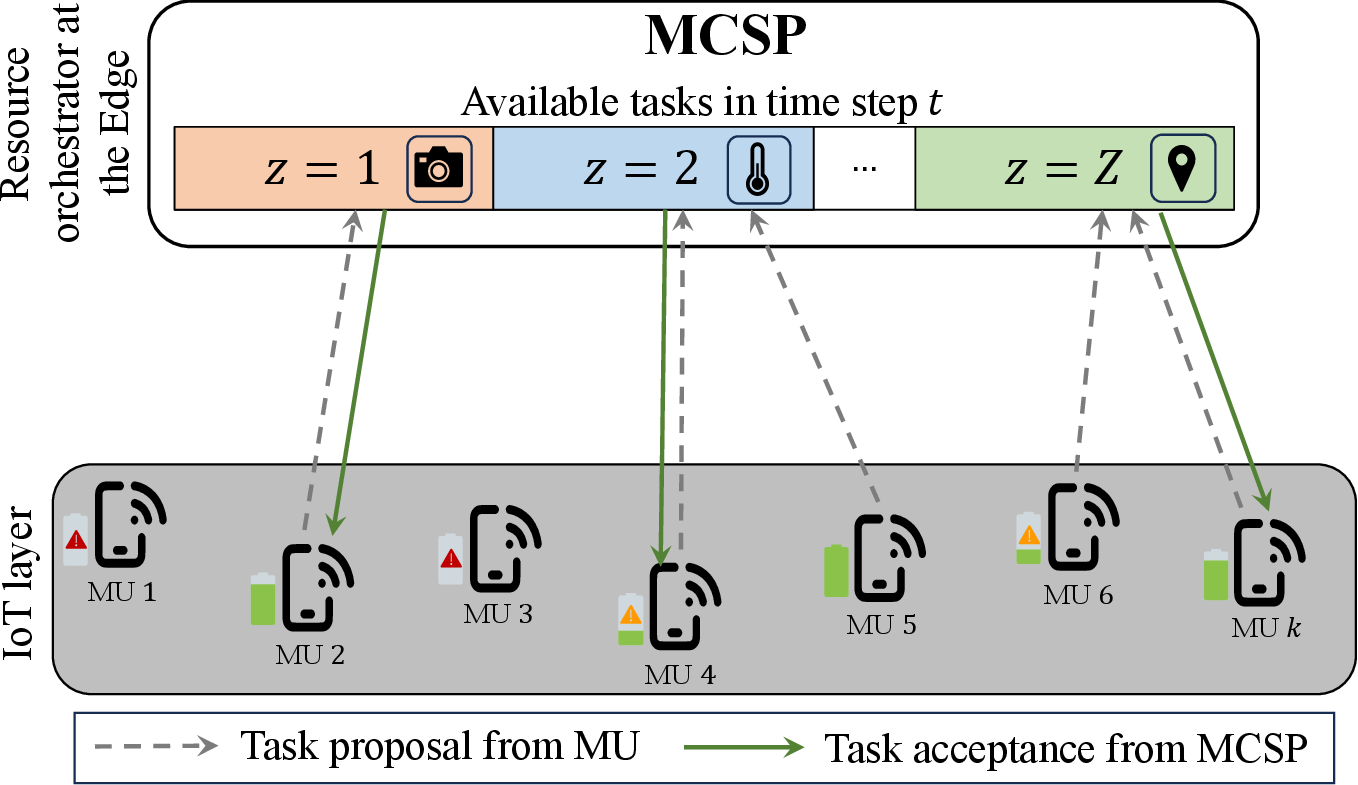}
\caption{The considered MCS system with heterogeneous MUs proposing to perform the available task types and an MCSP responding to the proposals. Here, no response from the MCSP indicates proposal rejection.}
\label{fig:systModel}
\end{figure}
We consider an \gls{mcs} system consisting of an \gls{mcsp} as a sensing resource orchestrator deployed at the edge/cloud and a set $\setOfMUs = \{\MUk\}_{k=1}^{K}$ of energy harvesting \glspl{mu}, as depicted in Fig. \ref{fig:systModel}.
The MUs are modeled within the IoT layer with integrated sensing, communication, and computing resources (ISCC).
Time is divided into discrete time steps of equal duration indexed by $t \in \{1, \ldots, T\}$.
In each time step $t$, the \gls{mcsp} publishes $N$ sensing tasks collected in a set $\setOfTasks_t=\{\taskWithIndex{n,t}\}_{n=1}^N$.
Every sensing task $\taskWithIndex{n,t}$ is categorized into $Z$ different task types denoted by $\setOfTaskTypes = \{z\}_{z=1}^Z$.
A task type $z$ may represent, for example, temperature sensing, environmental sensing, or traffic monitoring tasks.
All the tasks of the same type are collected in the set $\setOfTasksWithTypeZ\subseteq\setOfTasks_t$.
Each task type $\taskTypeZ$ is characterized by the average size $\taskSize$ of the sensing result in bits.
We assume that each published task $\taskWithIndex{n,t}$ requires only one \gls{mu} to complete it.
The \gls{mcsp} can publish multiple tasks of the same task type $\taskTypeZ$ if the \gls{dr} requires multiple sensing samples.

\subsection{Mobile Units (MUs)}
In every time step $t$, $\MUk$ decides whether to perform its preferred task of type $\taskTypeZ$ in $\setOfTasksWithType{z}$, or to remain idle to save its energy for potentially more rewarding future tasks.
If it decides to perform a task, $\MUk$ sends a task proposal $\sensingOffer$ to the \gls{mcsp} containing  the task type $\taskTypeIndex$ and the expected payment $\payment$ for its successful completion.
An \gls{mu} can submit at most one proposal per time step.
$\MUk$ calculates the desired payment $\payment$ based on its effort $J_{\MUindex,\timeindex}^{\taskTypeIndex}$ required to complete a task of type $\taskTypeZ$.
$J_{\MUindex,\timeindex}^{\taskTypeIndex}$ is measured in monetary units and depends on the time $\totalTime$ and energy $\totalEnergy$ required to complete the task as 
$J_{\MUindex,\timeindex}^{\taskTypeIndex}=\alpha\totalTime+\beta\totalEnergy$,
where $\alpha$ and $\beta$ are importance factors for time and energy efforts, respectively.
The time $\totalTime$ is given by $\totalTime = \sensingTime + \computationTime + \communicationTime$, where $\sensingTime$, $\communicationTime$ and $\computationTime$ are the sensing, computation and communication times, respectively.
$\sensingTime$ indicates the time required to sense and produce valid sensing data of size $d_z$ measured in bits.
It is drawn from a random distribution with probability density function $f^{\taskTypeZ}_{\sensingTime}$.
The expected value $\expectedSensingTime =\mathbb{E}[\sensingTime]$ of the sensing time depends on the task type $\taskTypeZ$ and the capabilities of $\MUk$.
Each \gls{mu} processes the sensing data to generate the sensing result~\cite{huang2022timedependent}.
The time $\computationTime$ required for the processing is calculated as $\computationTime= {c_z d_z}/{f^{\text{local}}_k}$, where $c_z$ is a random variable modeling the processing complexity of task type $\taskTypeZ$ and $f^{\text{local}}_k$ is the CPU frequency of $\MUk$ measured in Hz.
After processing, the final sensing result $r_{k,n,t}$, with size $\taskSize < d_z$ bits, is transmitted back to the \gls{mcsp}.
The transmission time $\communicationTime$ depends on the stochastic characteristics of the channel between $\MUk$ and \gls{mcsp}.

To perform a task, $\MUk$ spends energy $\totalEnergy$ given by $\totalEnergy = p_k^{\mathrm{sense}}\sensingTime + p_{k}^{\mathrm{comm}} \communicationTime  + \compPower \computationTime,
$ where  $\txPower$ is $\MUk$'s transmit power and $\compPower$ is the power required to process task $\taskWithIndex{n,t}$.
To perform the tasks, every $\MUk$ uses energy from its battery with capacity $B_\mathrm{max}$.
The battery status $b_{k,t}$ in time step $t$ is $b_{k,t} = \min(b_{k,t-1}-\totalEnergy+\totalHarvestedEnergy, B_\mathrm{max})$, where $\smash[b]{\totalHarvestedEnergy}$ is the amount of energy harvested by MU $k$.
The utility or profit of $\MUk$ in time step $\timeindex$ is
\begin{align}
\label{eq:utilityMU}
    \utilityMU = \payment - J_{\MUindex,\timeindex}^{\taskTypeIndex}.
\end{align}
The $\MUk$ must know the exact time and energy efforts it requires to perform task $\taskWithIndex{n,t}$ to calculate $\utilityMU$.
However, given the random nature of $\totalTime$ and $\totalEnergy$, the \glspl{mu} cannot know the exact efforts required before performing the task.
Thus, $\MUk$ estimates the expected utility
\begin{equation}
    \label{eq:expectedUtilityMU}
    \expectedUtilityMU =  \frac{1}{T}\sum_{t=1}^T\mathbb{E}[\utilityMU],
\end{equation}
where $\taskWithIndex{n,t} \in \setOfTasksWithTypeZ$.
This estimate is then used to make a task proposal decision and to select the payment $\payment$.
We define $\mathcal{I}^{\mathrm{MU}}_{k} = \{ \expectedUtilityMU | \taskTypeIndex\} $ as the \textit{MU-side} information.
Note that $\mathcal{I}^{\mathrm{MU}}_k$ is not available in advance at $\MUk$ and has to be learned over time by performing tasks.
Based on its own estimates of $\mathcal{I}^{\mathrm{MU}}_k$, every $\MUk$ sends a proposal $\sensingOffer$ to the \gls{mcsp}.

\subsection{Mobile Crowdsensing Platform}
The \gls{mcsp} receives the set of task proposals $\mathcal{O}_t$ from all \glspl{mu} for the available tasks in $\setOfTasks_t$ and decides which $\MUk$ executes which task $\taskWithIndex{n,t}$.
The task assignment decision is denoted by $x_{k,n,t}\in\{0,1\}$. $x_{k,n,t}=1$ means $\taskWithIndex{n,t}$ is assigned to $\MUk$, otherwise $x_{k,n,t}=0$.
The \gls{mcsp} can assign the task $\taskWithIndex{n,t}$ to only one of the \glspl{mu} proposing to that type of task, i.e.,
$\sum_{k}x_{k,n,t}\leq 1 \text{ } \forall n,t$.
The \gls{dr} pays the \gls{mcsp} for every executed task.
The earning $\rewardTaskCompletion$ which the \gls{mcsp} gets when $\taskWithIndex{n,t} \in \setOfTasksWithTypeZ$ is performed by $\MUk$ depends on the task type $\taskTypeZ$ and the quality factor $q_{k,n,t}\in[0,1]$ of the sensing result $r_{k,n,t}$.
$q_{k,n,t}$ is calculated using a quality function $Q_z$ as
$q_{k,n,t}=Q_z(r_{k,n,t})$.
These quality functions $Q_z$ could be, e.g., \gls{psnr} for images or accuracy of the temperature sensing.
The \gls{mcsp} and the \gls{dr} make a contractual agreement on the calculation of \gls{mcsp}'s earning 
\begin{equation}
\rewardTaskCompletion = (1 + q_{k,n,t})\rewardPerTaskType,
\end{equation}
where $w_z$ is the minimum payment in monetary units the \gls{mcsp} charges the \gls{dr} for performing a task of type $z$.
As the quality $q_{k,n,t}$ of the sensing result $r_{k,n,t}$ is unknown to the \gls{mcsp} in advance, $\rewardTaskCompletion$ is also unknown.
The utility $\utilityTask$ of the \gls{mcsp} when assigning $\MUk$ to $\taskWithIndex{n,t} \in \setOfTasksWithTypeZ$ is
\begin{align}
\label{eq:utilityMCSP}
    \utilityTask = (\rewardTaskCompletion - \payment). 
\end{align}
The \gls{mcsp} can maximize its utility $\utilityTask$ by balancing the quality of the sensing result of $\MUk$ and its payment.
However, as the \gls{mcsp} does not know the quality factor $q_{k,n,t}$ of the \gls{mu} performing task $\taskWithIndex{n,t} \in \setOfTasksWithTypeZ$, it estimates the expected utility when assigning $\MUk$ to a task type $\taskTypeZ$
\begin{equation}
    \label{eq:expectedUtilityMCPS}
    \expectedUtilityTask = \frac{1}{T}\sum_{t=1}^T\mathbb{E}[\utilityTask].
\end{equation}
We define $\mathcal{I}^{\mathrm{MCSP}}_{\taskTypeIndex} = \{ \expectedUtilityTask| k \} $ as the \textit{\gls{mcsp}-side} information about the \glspl{mu}.
$\mathcal{I}^{\mathrm{MCSP}}_{\taskTypeIndex}$ contains information about \gls{mcsp}'s income and the payments for all \glspl{mu}. $\mathcal{I}^{\mathrm{MCSP}}_\taskTypeIndex$ is not readily available at the \gls{mcsp} and is learned over time from experience gained from selecting \glspl{mu}.

The combination of \gls{mu}-side and \gls{mcsp}-side information, denoted by $\mathcal{I} = \{ \mathcal{I}^{\mathrm{MU}}_\MUindex, \mathcal{I}^{\mathrm{MCSP}}_\taskTypeIndex | \MUindex,\taskTypeIndex\} $, is the \textit{complete} information and is unknown to the \glspl{mu} and the \gls{mcsp} in advance.

\section{Dynamic Matching Game Formulation} \label{sec:problem_formulation}
The considered \gls{mcs} scenario is a two-sided market in which the \gls{mcsp} orchestrates sensing resources to execute tasks and the \glspl{mu} offer their sensing resources in exchange for a payment as incentives~\cite{shapley1971assignment}.
We assume that the \gls{mcsp} and the \glspl{mu} are rational and independent entities which take their own decisions based on their preferences.
Since their preferences influence their respective utilities, we use matching theory~\cite{Gu2015matching}, specifically dynamic matching theory, to analyze and solve the joint task proposal and task assignment problem.
Matching theory aims to obtain a \emph{stable matching}, i.e., to find task allocations where the \glspl{mu} and the \gls{mcsp} cannot improve their individual utilities by changing the assignment.
However, due to the limited battery capacities of the \glspl{mu}, their availability depends on the amount of energy available.
This means the number of available \glspl{mu} changes over time which makes the game dynamic and naturally hard to solve.
The dynamic matching problem captures the key challenge in service orchestration, where sensing services must be allocated under uncertainty, energy constraints, and competition, while maintaining long-term system efficiency.

We denote our dynamic matching game in time step $t$ as  $\mathcal{G}_t$.
The \glspl{mu}' preference ordering $\MUpreferenceTimeDependent$ ranks the task types $\taskTypeZ\in\setOfTaskTypes$ w.r.t. the achieved expected utility over the remaining time horizon, i.e., $\taskTypeZ \MUpreferenceTimeDependent \taskTypeZ'$, iff,
\begin{equation}
\label{eq:MU_preference}
\sum_{t'=t}^T \mathbb{E}\!\left[U^{\mathrm{MU}}_{\MUindex,n,t'}\!\mid\!\taskWithIndex{n,t}\!\in\!\setOfTasksWithType{z}\right]
\ge
\sum_{t'=t}^T \mathbb{E}\!\left[U^{\mathrm{MU}}_{\MUindex,m,t'}\!\mid\!\taskWithIndex{m,t}\!\in\!\setOfTasksWithType{z'}\right].
\end{equation}
This means, $\MUk$ prefers task type $\taskTypeZ$ over $\taskTypeZ'$ if the utility of performing tasks of type $\taskTypeZ$ is higher than the utility of tasks of type $\taskTypeZ'$ considering the remaining time horizon.
Similarly, the \gls{mcsp} prefers \glspl{mu} which yield the highest expected utility $\expectedUtilityTask$ for each task type $\taskTypeZ$, i.e., 
\begin{align}
    \MUwithIndex{k} \Taskpreference{}\MUwithIndex{l} \iff \expectedUtilityTaskWithIndex{k,\taskTypeIndex} \geq \expectedUtilityTaskWithIndex{l,\taskTypeIndex}.
    \label{eq:MCSP_preference}
\end{align}
As a result, for the assignment of a task of type $\taskTypeZ$, the \gls{mcsp} prefers $\MUk$ over $\MUwithIndex{l}$ if $\MUk$ provides higher utility compared to $\MUwithIndex{l}$.
This preference ranking can only be correctly determined with \gls{mcsp}-side information $\mathcal{I}^{\mathrm{MCSP}}$.

The task proposal and task assignment game $\mathcal{G}_t$ in time step $t$ is a tuple $\mathcal{G}_t = (\setOfMUs_t, \{\mathrm{MCSP}\}, \setOfTasks_t, \MUpreferenceTimeDependent, \Taskpreference)$, where $\setOfMUs_t \subseteq \setOfMUs$ is a subset of total \glspl{mu} available in time step $t$.
$\MUk \in \setOfMUs_t$ signals its willingness to participate in any task of type $\taskTypeZ$ by sending a task proposal $\sensingOffer$ to the \gls{mcsp}. 
Based on the proposals, the \gls{mcsp} performs the task assignment according to $\Taskpreference{}$.
The task assignment decisions for all \gls{mu}s and available tasks in $t$ are collected in the matrix $\assignmentMatrix$.

For two \glspl{mu}, $\MUwithIndex{k}$ and $\MUwithIndex{l}$, and two tasks, $\taskWithIndex{n,t}$ and $\taskWithIndex{m,t}$ in time step $t$.
The pair $(\MUwithIndex{k}, \taskTypeZ')$ is called a blocking pair, if both, $\MUk$ and the \gls{mcsp} can improve their utilities by deviating from the current assignment~\cite{WSaad_DynamicMatching_2018}.
The existence of the blocking pair $(\MUwithIndex{k}, \taskTypeZ')$
causes the matching 
$\assignmentMatrix$
to be unstable because $\MUwithIndex{k}$ could switch to $\taskWithIndex{m,t} \in \setOfTasksWithType{z'}$ and both, the $\MUwithIndex{k}$ and the task $\taskWithIndex{m,t}$ would obtain a more efficient matching and therefore a higher utility.
The assignment $\assignmentMatrix$ is said to be dynamically stable if no such blocking pairs exist~\cite{WSaad_DynamicMatching_2018}.
In MCS, this means that each \gls{mu} is assigned to its most preferred task while the \gls{mcsp} selects its most preferred \gls{mu} for each task.
This assignment should maximize the individual utilities of the \glspl{mu} and the \gls{mcsp}, and consequently, the social welfare.
The performance of the whole MCS system is directly affected due to \emph{collisions}, i.e., the difference between the total number of proposing MUs and the number of assigned MUs for each task type.
A dynamically stable solution does not have any collisions.

\section{Energy-aware Two-Sided Learning Algorithm}
\label{sec:algorithm}
\subsection{Overview}
To optimally solve the game $\mathcal{G}$ formulated in Sec. \ref{sec:problem_formulation}, complete information $\mathcal{I}$ is required at the \glspl{mu} as well as at the \gls{mcsp}, which is unrealistic to assume.
Considering every $\MUk$ can only learn its own MU-side information $\mathcal{I}^{\mathrm{MU}}_{k}$ and the \gls{mcsp} only learns the \gls{mcsp}-side information $\mathcal{I}^{\mathrm{MCSP}}_{\taskTypeZ}$, in this section, we present the Energy-aware Two-Sided Learning (ETSL) algorithm which aims to maximize the \glspl{mu}' and \gls{mcsp}'s utilities using this local information.
ETSL consists, i) an Energy-aware Task Proposal (ETP) algorithm, run locally by every \gls{mu}, and ii) Two-Sided Learning: Task Assignment (TSLTA), run by the \gls{mcsp}.

ETP is a Q-learning based solution implemented at every $\MUk$ to find an efficient task proposal strategy that maximizes the total achieved utility over the time horizon.
Each $\MUk$ considers the state $S_{k,t}=\langle b_{k,t} \rangle$, i.e., its battery and the action $A_{k,t}=\setOfTaskTypes \cup \{-1\}$, where $-1$ is the idle action.
The state space in this work is deliberately minimal, comprising of only the battery state of the respective \gls{mu}. As we assume a quasi-static channel model and the task arrivals are i.i.d., the battery level constitutes a sufficient statistic for the task proposal decision, in the sense that no additional observable variable alters the transition probabilities or the expected reward. A richer state representation would therefore increase the dimensionality of the problem without improving the quality of the resulting policy.

Every $\MUk$ receives instantaneous utility $R_{k,t}=\utilityMU$ as a reward.
Using ETP, each $\MUk$ independently learns the effort $J_{k,t}^\taskTypeZ$ required to perform task $\taskWithIndex{n,t}$ of type $\taskTypeZ$ by proposing and performing tasks of different types.
This is crucial in a scenario where the \glspl{mu} have limited energy. Moreover, the task proposal decision at time step $t$ not only impacts their performance, but also their ability to propose in the next time steps.
The estimation of the expected efforts $\hat{J}_{k,t}^\taskTypeZ$ for $\MUk$ is only possible when it gets assigned to tasks of type $\taskTypeZ$.

\begin{figure}[t]
\begin{algorithm}[H]
\caption{Energy-aware Task Proposal (ETP) algorithm}
\label{alg:proposed-algorithm}
\scriptsize
    \begin{algorithmic}[1]
        \STATE Initialize $Q(S_{k,t},A_{k,t}) = 0$, for all states $S_{k,t} \in \mathcal{S}_k$ and actions $A_{k,t} \in \mathcal{A}_k$, $\estimatedexpectedUtilityMUWithIndex{k,n,t}, \hat{J}^{z}_{k,t}\quad \forall k \in \mathcal{K}, z \in \mathcal{Z}, t \in \{1,\ldots,T\}$ and $\epsilon_1=\epsilon_\mathrm{max}$.
        \FOR{$t = 1, \ldots, T$}
            \STATE Observe current available tasks and current state $S_{k,t}=\langle b_{k,t}\rangle $.
            \STATE Obtain a feasible set $\mathcal{A}^\mathrm{f}_{k,t}$.
            \IF{$\eta < \epsilon_t$ where $\eta\sim\mathcal{U}[0,1]$}
                \STATE Randomly select a task type $z\in\mathcal{Z}$ from $\mathcal{A}^\mathrm{f}_{k,t}$. \algorithmiccomment{Exploration}
            \ELSE
                \STATE Select $A_{k,t}=\argmax_{A_{k,t}\in \mathcal{A}^\mathrm{f}_{k,t}} Q(S_{k,t}, A_{k,t})$. \algorithmiccomment{Exploitation}
            \ENDIF
            \STATE If $A_{k,t}=-1$, $\MUk$ stays idle.
            \STATE If $A_{k,t}=z\in\mathcal{Z}$, select payment $\hat{P}_{k,z} \leftarrow P_{\text{effort}}(\hat{J}^{z}_{k,t})$.
            \STATE Send task proposal $\sensingOffer$.
            \STATE Wait for the MCSP's decision $x_{k,n,t}$ from TSLTA \cite{Dongare_TSL_2024_ICC}.
            \IF{$\elementOfAssignmentMatrixWithIndex{k,n,t}=1$}
                \STATE Perform the task $\taskWithIndex{n,t}$ and transmit the result $r_{k,n,t}$ to \gls{mcsp}.
                \STATE Receive payment $\MUpaymentProposal$ and observe $\utilityMU$.
                \STATE Update estimates $\estimatedexpectedUtilityMUWithIndex{k,n,t}$ and $\hat{J}_{k,t}^\taskTypeIndex$ based on $\utilityMU$ and ${J}_{k,t}^\taskTypeIndex$.
                \STATE Update the $Q(S_{k,t}, A_{k,t})$. \algorithmiccomment{Q-learning update rule}
            \ELSE
                \STATE $\estimatedexpectedUtilityMUWithIndex{k,n,t} \xleftarrow[]{}  \estimatedexpectedUtilityMUWithIndex{k,n,t-1}$, $\hat{J}_{k,t}^\taskTypeIndex \xleftarrow[]{}  \hat{J}_{k,t-1}^\taskTypeIndex$.
            \ENDIF
            \STATE Update $\epsilon_{t+1} = \min\{\epsilon_{\mathrm{min}}, \epsilon_t*\theta\}$
        \ENDFOR
    \end{algorithmic}
\end{algorithm}
\vspace{-4mm}
\end{figure}
ETP is summarized in Alg. \ref{alg:proposed-algorithm}.
Every $\MUk$ starts by initializing all Q-values $Q(S_{k,t}, A_{k,t})$ to zero for $t=1$, for all states $S_{k,t}$, and actions $A_{k,t}$.
For every task type $z$, each $\MUk$ maintains an initial estimate of the task effort $\hat{J}_{k,0}^z$ (line 1).
In every time step $t$, $\MUk$ observes the current state $b_{k,t}$ (line 3).
Due to the energy constraints, $\MUk$ collects all the feasible tasks in time step $t$ into the set $\mathcal{A}_{k,t}^\mathrm{f}$, i.e., if the expected energy $\Bar{E}_{k,n,t}$ required to perform the task $\taskWithIndex{n,t}\in\setOfTasksWithTypeZ$ is smaller than the battery state $b_{k,t}$, then task type $\taskTypeZ$ is added to the set $\mathcal{A}_{k,t}^\mathrm{f}$ (line 4).
To balance exploration and exploitation of the Q-learning algorithm, $\MUk$ follows a decaying $\epsilon$-greedy policy.
With probability $\epsilon_t$, $\MUk$ explores by randomly selecting a task type $\taskTypeZ$ from the feasible set $\mathcal{A}_{k,t}^\mathrm{f}$ (line 6).
With a probability $1-\epsilon$, it exploits its current knowledge by selecting the task of type $\taskTypeZ$ from the feasible set $\mathcal{A}_{k,t}^\mathrm{f}$ that maximizes the estimated expected utility $\estimatedexpectedUtilityMUWithIndex{k,n,t}$, based on the Q-values learned from previous time steps (line 8).
If idle action $A_{k,t}=-1$ is selected, the $\MUk$ remains idle in time step $t$.
If $A_{k,t}=z\in\mathcal{A}^\mathrm{f}_{k,t}$ is selected, $\MUk$ calculates the payment $\hat{P}_{k,z}$ based on its estimated task effort $\hat{J}_{k,t-1}^z$ (line 11).
$\MUk$ then sends a task proposal $\sensingOffer$, including the selected task type and payment, to the \gls{mcsp} (line 12).
The \gls{mcsp} collects all incoming proposals and decides which \glspl{mu} are assigned to the available tasks (line 13). 
If $\MUk$ is assigned to the task $\taskWithIndex{n,t}$, i.e. $x_{k,n,t}= 1$, it proceeds to execute the task (line 14).
If the task is successfully completed, $\MUk$ receives the payment $\hat{P}_{k,z}$ and observes the utility $\utilityMU$ as well as the exact task effort $J_{k,t}^z$ (line 15).
$\utilityMU$ is used to update the estimated expected utility $\estimatedexpectedUtilityMUWithIndex{k,n,t}$ as
$\estimatedexpectedUtilityMUWithIndex{k,n,t} = \estimatedexpectedUtilityMUWithIndex{k,n,t-1}+\frac{(\utilityMU-\estimatedexpectedUtilityMUWithIndex{k,n,t-1})}{N_k^\taskTypeZ}$, where $N_k^\taskTypeZ$ represents the number of times $\MUk$ was assigned to a task of type $\taskTypeZ$ (line 17).
The effort estimate $\hat{J}_{k,t}^\taskTypeZ$ is updated similarly (line 17).
The achieved utility $\utilityMU$ is used as the reward to update the Q-value for the selected state-action pair using the standard Q-learning update rule (line 18).
If $\MUk$ is rejected by the \gls{mcsp}, the estimated task efforts and utility remain unchanged (line 20).
By learning the task efforts for the available task types and acceptance probabilities of the \gls{mcsp}, every \gls{mu} refines its task proposal strategy, improving its ability to balance energy constraints with maximizing utility.
ETP uses the idle action smartly such that the \gls{mu} can save its energy if it does not find a feasible task to propose.

At the \gls{mcsp}, the task assignment decisions $x_{k,n,t}$ are made in every time step $\timeindex$ based on the proposals $\sensingOffer$ received from the \glspl{mu}.
The MCSP cannot observe MU battery dynamics and therefore cannot anticipate the long-term impact of its task assignment decisions on MU availability.
Under this information structure, any MCSP-side learning strategy is inherently myopic, making TSLTA from~\cite{Dongare_TSL_2024_ICC} an appropriate solution for instantaneous utility maximization.

In TSLTA, the \gls{mcsp} initializes the $\estimatedexpectedUtilityTaskWithIndexMCSP{k,n,t}$ for all \glspl{mu} and all task types.
The set $\setOfMUs_t$ of MUs is the action space for TSLTA.
The \gls{mcsp} may assign the task to a random \gls{mu} proposing for the task $\taskWithIndex{n,t}\in\setOfTasksWithType{z}$ to improve its estimate of $\estimatedexpectedUtilityTaskWithIndexMCSP{k,n,t}$; or, it may exploit the available knowledge to assign the task to \glspl{mu} which maximize its estimated expected utility $\estimatedexpectedUtilityTaskWithIndexMCSP{k,n,t}$.
To balance the exploration-exploitation, the \gls{mcsp} uses $\epsilon$-greedy action selection.
The assignment decisions are sent back individually to each proposing $\MUk$.
After this, the assigned \glspl{mu} perform the task and transmit the sensing result $r_{k,n,t}$ back to the \gls{mcsp}.
Using $r_{k,n,t}$, the \gls{mcsp} evaluates the quality of the result $q_{k,n,t}$ and observes the $\utilityTask$.
Then the \gls{mcsp} updates $\estimatedexpectedUtilityTaskWithIndexMCSP{k,\taskTypeZ,t}= \estimatedexpectedUtilityTaskWithIndexMCSP{k,\taskTypeZ,t-1} + \frac{(\utilityTask-\estimatedexpectedUtilityTaskWithIndexMCSP{k,n,t-1})}{N_k^z}$, where $N_k^z$ is the number of times $\MUk$ is assigned to task type $\taskTypeZ$.
With the help of ETP and TSLTA, we jointly optimize the \glspl{mu}' task proposals to maximize their utilities as well as \gls{mcsp}'s task assignments to maximize the \gls{mcsp}'s utility.

\subsection{Convergence of ETSL}

Owing to the coupled, non-stationary nature of the two-sided learning problem, we characterize convergence by conditioning
on stationarity of one side at a time, following the standard alternating-freeze approach used in multi-agent reinforcement learning analysis~\cite{watkins1992q,jaakkola1994convergence}.

\begin{assumption}[Stationary MCSP policy]
\label{as:mcsp_stationary}
The MCSP's assignment rule has converged to the greedy form
$k^*(z,t) = \arg\max_{k \in \mathcal{P}_{z,t}} \tilde{U}^{\mathrm{MCSP}}_{k,z,t-1}$,
held fixed for $t \geq t_0$.
\end{assumption}

\begin{proposition}[MU-side convergence]
\label{prop:etp_convergence}
Under Assumption~\ref{as:mcsp_stationary}, the decision problem faced by $\mathrm{MU}_k$ is a stationary MDP $\mathcal{M}_k = (\mathcal{S}_k, \mathcal{A}_k, P_k, R_k)$ with $S_{k,t} = \langle b_{k,t} \rangle$. If (i) rewards $R_{k,t}$ are bounded, (ii) the learning-rate schedule satisfies the Robbins-Monro conditions, and (iii) the exploration policy is
greedy in the limit with infinite exploration (GLIE), then $Q(S_{k,t}, A_{k,t}) \to Q^*(S_{k,t}, A_{k,t})$ with
probability~1.
\end{proposition}

\begin{IEEEproof}
Under Assumption~\ref{as:mcsp_stationary}, the acceptance probability of $\mathrm{MU}_k$'s proposals depends only on the competition from other MUs, so the transition kernel $P_k(S_{k,t+1} \mid S_{k,t}, A_{k,t})$ is stationary.
The battery evolution depends on $A_{k,t}$ only through the now stationary acceptance probability, and $E_{k,n,t}$, $E^h_{k,t}$ are themselves stationary.
The result follows directly from the classical Q-learning convergence theorem~\cite{watkins1992q,jaakkola1994convergence}.
\end{IEEEproof}

\begin{assumption}[Approximately stationary aggregate proposal process]
\label{as:mu_stationary}
For each task type $z$, let $M_{z,t} := |\mathcal{P}_{z,t}| = \sum_{k=1}^{K} \mathbbm{1}\{A_{k,t}=z\}$ denote the aggregate number of proposals received at time $t$.
We assume that for $t \geq t_0$, the distribution of $M_{z,t}$ is approximately stationary, even though individual MUs' proposal decisions $A_{k,t}$ may remain non-stationary due to ongoing battery-state fluctuations $b_{k,t}$.
\end{assumption}

\textit{Justification:} Battery dynamics evolve independently across MUs, driven by independent energy-harvesting processes (see Sec.~\ref{sec:numerical_evaluation}).
Even after $Q(S_{k,t}, A_{k,t})$ has converged (Proposition~\ref{prop:etp_convergence}), individual proposal sequences $\{A_{k,t}\}_t$ remain stochastic due to slow mixing of $b_{k,t}$.
However, since $M_{z,t}$ aggregates $K$ approximately independent indicator variables, a mean-field argument gives
\begin{equation}
\frac{M_{z,t}}{K} \xrightarrow{\text{a.s.}} \bar{\rho}_z
\quad \text{as } K \to \infty,
\end{equation}
with fluctuations of order $O(1/\sqrt{K})$ by the law of large numbers.
For large number of MUs used in the evaluation, this concentration is substantial, supporting the treatment of the competition landscape $M_{z,t}$, which determines the proposal pool $\mathcal{P}_{z,t}$ from which TSLTA selects as approximately stationary for $t \geq t_0$, independent of MU-specific non-stationarity.

\begin{proposition}[MCSP-side convergence]
\label{prop:tslta_convergence}
Under Assumption~\ref{as:mu_stationary}, TSLTA reduces, independently for each task type $z$, to a stochastic multi-armed bandit problem over arms $k \in \mathcal{P}_{z,t}$ with stationary mean reward $U^{\mathrm{MCSP}}_{k,z}$.
Given bounded rewards and $N^z_k \to \infty$ a.s.\ for all $k$, the estimate $\tilde{U}^{\mathrm{MCSP}}_{k,z,t} \to U^{\mathrm{MCSP}}_{k,z}$ a.s., and the induced greedy policy converges to $k^*(z) = \argmax_k U^{\mathrm{MCSP}}_{k,z}$.
\end{proposition}

\begin{IEEEproof}
The assignment problem decouples across task types $z$. Within a fixed $z$, stationarity of $\rho_{k,z}$ renders the proposal arrival process stationary, and the incremental update in~(5) is exactly the sample-mean estimator, whose convergence follows from the strong law of large numbers and standard $\epsilon$-greedy bandit results~\cite{sutton2018reinforcement}.
\end{IEEEproof}

\textit{Remark:} Propositions~\ref{prop:etp_convergence} and~\ref{prop:tslta_convergence} characterize the fixed point each side converges to under stationarity of the other, because both sides learn simultaneously in practice, this is only approximately satisfied at finite $t$. The two updates operate on different effective timescales ($\tilde{U}^{\mathrm{MCSP}}_{k,z,t}$ updates once per assignment, while $Q(S_{k,t}, A_{k,t})$ updates once per MU decision), which is consistent with the two-timescale stochastic approximation framework of~\cite{borkar1997stochastic}; we leave a full joint convergence proof to future work and instead provide extensive empirical evidence of stability in Sec.~V.

\subsection{Complexity and overhead of ETSL}
ETP has linear complexity in the number of feasible actions, i.e., 
$O(|\mathcal{A}^\mathrm{f}_{k,t}|)$ for each \gls{mu} and its current battery state, with worst-case complexity of $O(Z)$.
Note that the Q-table grows with linear space complexity $O(LZ)$ for a fixed number of battery levels $L$.
Likewise, TSLTA has linear complexity in the number of proposals, i.e., $O(|\mathcal{O}_t|)$, with worst-case complexity of $O(K)$ and space complexity of $O(KZ)$.
Note that for both, the \gls{mcsp} and the \glspl{mu}, the communication overhead required for matching is low.
The \glspl{mu} send task offers to the \gls{mcsp} which contains only the task type and the payment information.
The task acceptance and the task rejection messages which the \gls{mcsp} transmits back to the respective \glspl{mu} are also very short.

\section{Numerical Evaluation}
\label{sec:numerical_evaluation}
\begin{figure*}[!t]
    \centering
    \includegraphics[width=0.75\textwidth]{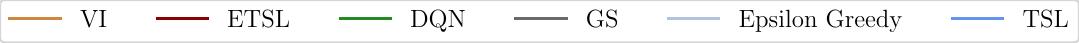}\\[-2mm]

    \subfloat[System performance analysis\label{fig:socialWelfare}]{%
        \includegraphics[width=0.32\textwidth]{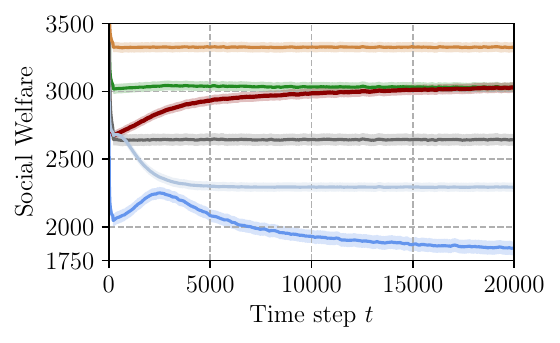}}
    \hfil
    \subfloat[Energy consumption analysis\label{fig:energyConsumption}]{%
        \includegraphics[width=0.32\textwidth]{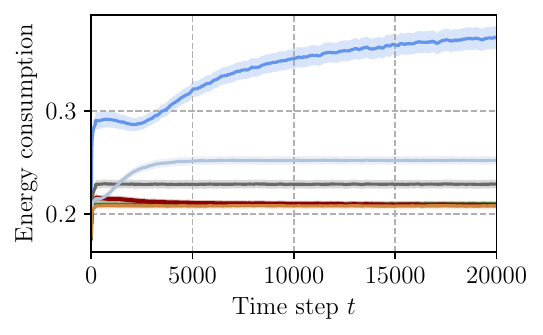}}
    \hfil
    \subfloat[Competition analysis\label{fig:collisionRatio}]{%
        \includegraphics[width=0.32\textwidth]{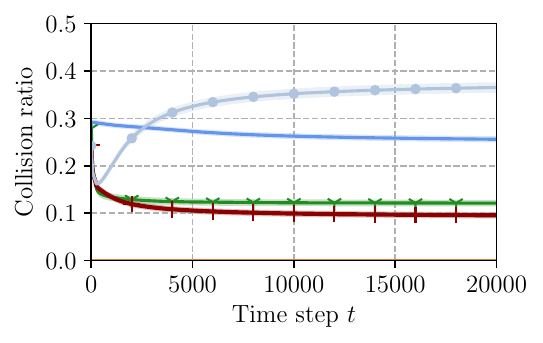}}
    \\[-1mm]

    \subfloat[MU scalability analysis\label{fig:scalability}]{%
        \includegraphics[width=0.24\textwidth]{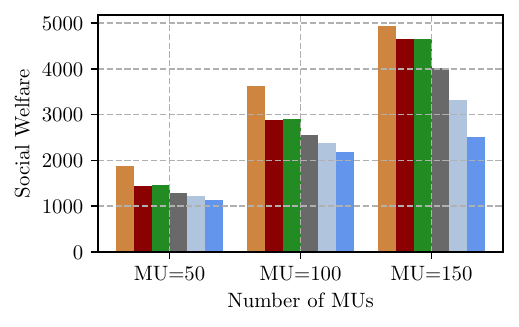}}
    \hfil
    \subfloat[Task scalability analysis\label{fig:scalability2}]{%
        \includegraphics[width=0.24\textwidth]{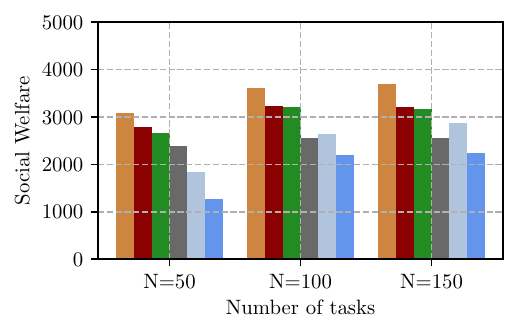}}
    \hfil
    \subfloat[Sensitivity: EH dynamics\label{fig:eh_sensitivity}]{%
        \includegraphics[width=0.24\textwidth]{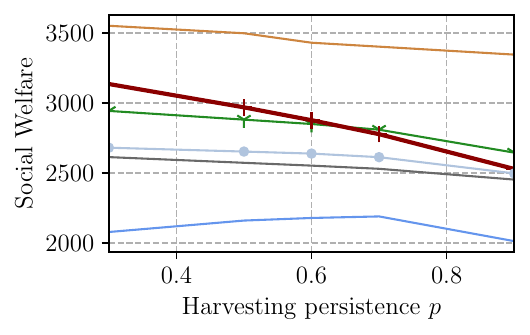}}
    \hfil
    \subfloat[Sensitivity: battery capacity\label{fig:battery_sensitivity}]{%
        \includegraphics[width=0.24\textwidth]{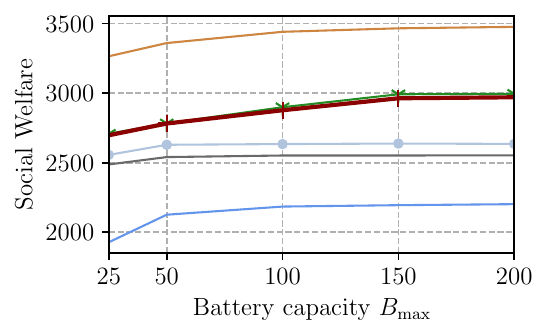}}

    \caption{Performance comparison of the ETSL under different metrics and analyses.}
    \label{fig:all_plots}
\end{figure*}

\begin{table}[t]
\centering
\caption{Simulation Parameters}
\label{tab:SimPara}
\begin{tabular}{l l}
    \hline
    Sensing data size \cite{huang2022timedependent} & $\mathcal{U}[50, 100]$ Mbits \\
    Processed result size \cite{huang2022timedependent} & $\mathcal{U}[10, 20]$ Mbits \\
    Sensing time $\tau^{\text{sense}}_{k,z}$ \cite{huang2022timedependent} & $\mathcal{U}[60, 180]$ ms \\
    Communication time $\tau^{\text{comm}}_{k,z}$ \cite{Decentralized_Bernd_2024} & $\mathcal{U}[0.125, 0.5]$ s \\
    Local CPU frequency $f^{\text{local}} _{k}$ \cite{Mahn2021globalOrchestration} & $\mathcal{U}[1, 2]$ GHz \\
    Task processing complexity $c_z$ \cite{huang2022timedependent} & $\mathcal{U}[200, 300]$ \\
    MU battery capacity $B_\mathrm{max}$ \cite{Dongare_AoII_2024} & $100$ units \\
    Amount of energy harvested $E^{\text{h}}_{k,t}$ \cite{Dongare_AoII_2024} & $\mathcal{U}[0, 15]$ units \\
    \hline 
\end{tabular}
\end{table}
For the evaluation, we consider $300$ independent Monte Carlo iterations, each with $T=20000$ time steps.
In the baseline simulation, the number of available MUs is set to $K=100$, and the number of available task types $Z=10$.
The number of available tasks per time step varies between $[50, 100]$.
EH is modeled as a time correlated Markov chain with a harvesting ($\chi$) and a non-harvesting ($\Bar{\chi}$) state.
The transition probabilities are $P(\chi|\chi)=P(\Bar{\chi}|\Bar{\chi})=0.6$ and $P(\chi|\Bar{\chi})=P(\Bar{\chi}|\chi)=0.4$~\cite{Dongare_AoII_2024}.
For the different analyses such as MU and task scalability as well as the sensitivity to EH dynamics and the battery capacities, the respective parameters are mentioned exclusively.
Therefore, unless specified, the baseline scenario parameters are used in the analysis.

The learning rate is $\alpha=0.1$, the discount factor is $\gamma=0.9$, and the $\epsilon$-greedy exploration parameters are $\epsilon_{\max}=1$, $\epsilon_{\min}=0$, and $\theta=0.999$. The number of state levels is set to $L=21$.
All learning benchmark algorithms use the learning parameters which are determined through independent hyperparameter optimization.
The rest of the parameters along with their references are summarized in Table \ref{tab:SimPara}.

For comparison, centralized and decentralized state-of-the-art benchmarks are used. The centralized approaches assume complete information $\mathcal{I}$. Although such assumption makes them impossible to implement in real applications, they serve as theoretical baselines.\\
\textbf{Value Iteration (VI)}: This centralized dynamic programming-based approach computes the value of being in each battery state for each MU and identifies the optimal task proposal and task assignment strategies for \glspl{mu} and \gls{mcsp} to maximize their utilities using $\mathcal{I}$.\\
\textbf{Gale-Shapley (GS)~\cite{Gu2015matching}}: A centralized approach that provides a myopically stable solution to the dynamic task proposal and task assignment game by maximizing the individual utilities of the \glspl{mu} and the \gls{mcsp} in time step $t$ by using $\mathcal{I}$.

For the decentralized benchmark algorithms, the \gls{mcsp} uses the TSLTA algorithm for task assignment~\cite{Dongare_TSL_2024_ICC}.
At the \glspl{mu}' side, we use:\\
\textbf{TSLTP~\cite{Dongare_TSL_2024_ICC}}: Each \gls{mu} adopts decentralized gradient-based multi-armed bandit algorithm to learn its task proposal strategy. TSLTP does not account for \gls{mu}'s energy constraints.
Note that TSLTP and TSLTA together constitute the TSL.\\
\textbf{Deep-Q Network (DQN)}: A deep reinforcement learning baseline that operates on the same state, action, and reward spaces as the proposed ETP. DQN uses a neural network function approximator to estimate action-value functions.\\
\textbf{Epsilon-Greedy (EG)}: A lightweight baseline where each \gls{mu} learns task utilities using a multi-armed bandit algorithm with a decaying $\epsilon$-greedy exploration strategy and proposes tasks that maximize immediate expected profit.

We use achieved social welfare as a comparison metric to evaluate system-level performance of all the benchmarks as shown in Fig.\ref{fig:socialWelfare}.
The social welfare is the sum of the utilities of the MUs and the MCSP.
VI and GS utilize the complete information $\mathcal{I}$ to obtain matching solutions.
VI provides an upper bound by considering the impact of spending the energy resources on the future using dynamic programming, whereas, GS finds a myopically stable matching solution ignoring the consequences on the battery.
Our ETSL algorithm performs only $7.6\%$ lower than VI and outperforms GS by $13\%$.
This is because the ETSL helps the MUs to consider the future consequences of their task proposal decisions on their battery state, while obtaining an efficient task acceptance strategy for the MCSP based on the proposals.
By exploring the task types, the ETSL improves the estimate on task efforts and expected task rewards.
ETSL learns the trends in the \gls{mu} competition over time.
The results show that the EG and the TSL achieve $52.5\%$ and $64.9\%$ lower social welfare than our proposed approach, respectively.
Moreover, the performance of our proposed ETSL converges to that of the DQN because the considered state-action space of every MU is low-dimensional and discretized, for which tabular Q-learning provides sample-efficient and stable learning.
For a more complex scenario, DQN may exhibit better performance albeit at the cost of higher computational complexity.
The EG ignores the collisions while learning the task proposal strategy and thus results in a poor performance.
TSL ignores the energy consequences of the \glspl{mu} and thus performs well as long as energy is available but then deteriorates over time.

Fig. \ref{fig:energyConsumption} shows the total energy consumed per number of completed tasks.
Our proposed ETSL performs more tasks by smartly selecting tasks that offer higher utility while also allowing the MUs to stay idle when the available energy is scarce.
This is directly reflected in the results where the ETSL consumes on average at least $45.9\%$, $42.9\%$, $16.7\%$, and $9.1\%$ less energy than the EG, TSL, GS, and DQN algorithms respectively and converges to VI's performance.

In Fig. \ref{fig:collisionRatio}, we compare the number of collisions per total number of task proposals, also known as the collision ratio.
A lower collision ratio indicates that the \glspl{mu} have learned about the competition and are able to make task proposal decisions efficiently.
Collisions degrade the performance of the system as resources are wasted during a collision.
Note that GS and VI do not have any collisions as they exploit the complete information $\mathcal{I}$.
Our proposed ETSL algorithm achieves at least $73.7\%$, $54.5\%$, and $9.1\%$ lower collisions compared to the EG, TSL, and DQN.
This is because our ETSL constantly updates the task type preferences based on the current battery state, current competition, and the learned acceptance mechanism of the MCSP.
This helps the MUs to defer from proposing to some task types from which they are frequently rejected.
Evidently, ETSL is able to reduce the collisions over time and moves towards more stable allocations which enhances the overall performance and reduces resource wastage.
The $5$th–$95$th percentile confidence intervals in the plots above demonstrate the stability of ETSL's performance.

In Fig.~\ref{fig:scalability} and~\ref{fig:scalability2}, we analyze how the benchmark algorithms scale with increasing number of \glspl{mu} and increasing number of tasks, respectively.
We observe that the ETSL exhibits better scalability characteristics by achieving only $7.8\%$ lower social welfare than VI.
Moreover, ETSL converges to DQN in all cases, depicting its superior performance with the advantage of lower complexity than the DQN.
The other benchmark algorithms struggle to perform with higher number of \glspl{mu} which increases the competition.

Finally in Fig.~\ref{fig:eh_sensitivity}and Fig.~\ref{fig:battery_sensitivity}, we perform the sensitivity analysis for the EH process dynamics and the battery capacity of the \glspl{mu}.
In Fig.\ref{fig:eh_sensitivity}, we consider the probability of persistence, i.e., $P(\chi|\chi)=P(\Bar{\chi}|\Bar{\chi})$ between the range $0.3$ to $0.9$ in the steps of $0.1$.
As the persistence increases, the resulting Markov chain stays in the EH or non-EH state for a long time.
This reduced the resulting energy harvesting probability and thus results in a lower performance.
Here, the performance of the ETSL is consistent with the VI.

Similarly, in Fig.\ref{fig:battery_sensitivity}, we vary the battery capacity of the \glspl{mu} between $[25, 200]$.
As the battery capacity increases, the \glspl{mu} are able to store more energy and sustain longer in the environment for task proposals and executions.
Here, the proposed ETSL algorithm shows promising performance across the battery capacities which demonstrates the superior performance of the proposed ETSL algorithm.

\section{Conclusion}
\label{sec:conclusion}
This work addressed the joint task proposal and task assignment problem in \gls{mcs}-based SaaS for NGNs by formulating it as a dynamic two-sided matching game under incomplete information, where unknown sensing quality, task effort, and the energy limitations of battery-powered \glspl{mu} affect participant availability and sustainable sensing-resource orchestration. To tackle this challenge, we proposed ETSL, a fully decentralized and low-complexity learning framework that enables \glspl{mu} to learn energy-aware task proposal strategies while allowing the \gls{mcsp} to learn \gls{mu} sensing quality for efficient task assignment.
By explicitly incorporating energy constraints into the learning and matching process, ETSL supports sustainable system operation while accounting for the individual preferences and utilities of both \glspl{mu} and the \gls{mcsp}. Moreover, we use an alternating-freeze approach to show that the proposed ETSL solution converges to a stable matching solution. Simulation results confirm that ETSL outperforms state-of-the-art benchmark algorithms in terms of social welfare, energy efficiency, and collisions, while demonstrating its scalability and suitability for Sensing-as-a-Service in future NGNs.

\bibliographystyle{IEEEtran}
\bibliography{./bibliography/IEEEabrv,./bibliography/references}
\end{document}